\documentclass{article}
\usepackage{amsmath,amsfonts,amsthm,amssymb}
\usepackage{wasysym} 
\usepackage{bm}
\usepackage[usenames,dvipsnames]{color}
\usepackage{tikz-cd}
\usetikzlibrary{matrix,arrows,calc,positioning,shapes}
\usepackage{calc}
\usepackage{mathrsfs}
\usepackage{mathtools}
\usepackage[margin=2cm]{geometry}
\usepackage[colorlinks=true,linkcolor=blue,citecolor=red,urlcolor=Violet,bookmarksdepth=subsection,final]{hyperref}

\newtheorem{thm}{Theorem}[section]
\newtheorem{lem}[thm]{Lemma}

\newtheorem{cor}[thm]{Corollary}

\tikzset{ampersand replacement=\&}

\allowdisplaybreaks[1] 

  \newcommand{\oo}{\infty}
  \newcommand{\C}{\mathbb{C}}

  \newcommand{\A}{\mathcal{A}}
  \newcommand{\B}{\mathcal{B}}
  \newcommand{\D}{\mathcal{D}}
  \newcommand{\del}{\partial}
  \newcommand{\id}{\mathrm{id}}
  \newcommand{\eps}{\varepsilon}

\title{Reducing triangular systems of ODEs with rational coefficients,
with applications to coupled Regge-Wheeler equations}

\author{Igor Khavkine\\
	\texttt{khavkine@math.cas.cz}\\
	Institute of Mathematics, Czech Academy of Sciences,\\
	\v{Z}itn{\'a} 25, 115 67 Praha 1, Czech Republic}

\date{}

\begin{document}
\maketitle

\begin{abstract}
We concisely summarize a method of finding all rational solutions to an
inhomogeneous rational ODE system of arbitrary order (but solvable for
its highest order terms) by converting it into a finite dimensional
linear algebra problem. This method is then used to solve the problem of
conclusively deciding when certain rational ODE systems in upper
triangular form can or cannot be reduced to diagonal form by
differential operators with rational coefficients. As specific examples,
we consider systems of coupled Regge-Wheeler equations, which have
naturally appeared in previous work on vector and tensor perturbations
on the Schwarzschild black hole spacetime. Our systematic approach
reproduces and complements identities that have been previously found by
trial and error methods.
\end{abstract}

\section{Introduction} \label{sec:intro}

In the recent work~\cite{kh-vwtriang}, we have shown how, after a
separation of variables, the radial mode equations of the vector wave
equation $\square v_\mu = 0$ on the Schwarzschild black hole spacetime
may be significantly simplified by systematically decoupling them into an
upper triangular form, where the diagonal components are generalized
Regge-Wheeler operators and only a few of the off-diagonal components
are non-vanishing. The original radial mode equations constitute a
$4\times 4$ second order linear ODE, whose components are coupled in a
rather complicated way. The Regge-Wheeler operators appearing on the
diagonal of the upper triangular form are second order scalar
differential operators, with a well-studied spectral theory. This
simplification makes it possible to transfer that knowledge to the study
of the spectral theory of the original radial mode equations, which was
otherwise rather unapproachable.

What is remarkable is that the original equations, the simplified upper
triangular form, as well as the decoupling transformation are all
ordinary differential operators with rational coefficients (at least in
the standard Schwarzschild radial coordinate $r$). The existence of such
such a simplification, in particular the ability to set to zero most of
the off-diagonal terms in the upper triangular form, follows from
specific identities previously discovered by trial and error
in~\cite{berndtson}, and in part also in~\cite{rosa-dolan}
(see~\cite{kh-vwtriang} for a full discussion), and is certainly not
obvious \emph{a priori}. This naturally raises the questions of how
these identities could be recovered in a systematic approach, and
whether more could be discovered to push the simplifications described
above as far as possible. These questions become particularly relevant
for trying to repeat the same simplifications for the Lichnerowicz
equation $\square p_{\mu\nu} - 2 R_{(\mu}{}^{\lambda\kappa}{}_{\nu)}
p_{\lambda\kappa} = 0$, which has a role relative to linearized Einstein
equations analogous to that of the vector wave equation relative to
Maxwell equations. The relevant identities were also discovered
in~\cite{berndtson}, but only by means of voluminous trial and error
calculations and without a clear answer to whether they could be further
improved. We will revisit this point when considering examples in
Section~\ref{sec:examples}.

The above questions were left open in~\cite{kh-vwtriang} and are
answered in this work. The main systematic tool at our disposal is the
theory of rational solutions of ordinary differential equations (ODEs)
with rational coefficients. Under appropriate mild hypotheses on the
equation, the search for such solutions can be reduced to a finite
dimensional linear algebra problem (Theorem~\ref{thm:univ-mult}). We
summarize this theory in Section~\ref{sec:rat-sols}. The theory of
rational solutions of scalar rational ODEs is fairly well developed
(cf.\ the monograph~\cite{abramov} and the references therein; more
precise references are given in the text). Our innovation is to
synthesize this approach into an economical form, based on what we call
\emph{leading} (or \emph{trailing}) \emph{multipliers}, that is directly
applicable to our examples of interest, but also more generally to ODE
systems of arbitrary size and order. In Section~\ref{sec:offdiag}, we
consider the problem of setting to zero an off-diagonal block in an
upper triangular rational ODE system by a transformation with rational
coefficients. This problem is first reduced to an operator identity
(Equation~\eqref{eq:offdiag}), which in turn can be solved by converting
it into an inhomogeneous rational ODE system
(Theorem~\ref{thm:reduce-order}). Finally, in Section~\ref{sec:rw}, we
combine the results of Sections~\ref{sec:rat-sols} and~\ref{sec:offdiag}
to show how the special identities used in~\cite{kh-vwtriang}
and~\cite{berndtson} can be recovered with minimal effort, especially
when aided by computer algebra. In particular, we can conclusively
decide when simplifications of the kind described earlier do or do not
exist, with several examples given in Section~\ref{sec:examples}.
Section~\ref{sec:discuss} concludes with a discussion of the results and
an outlook to further work.

\section{Rational solutions of ODEs with rational coefficients}
\label{sec:rat-sols}

The main objects under our study will be ordinary differential operators
and equations (ODEs) with rational coefficients. We will usually denote the
independent variable by $r$. A differential operator $e$ applied to a
function $u = u(r)$ will be denoted by $e[u]$. Both $u$ and $e[u]$ could
be vector valued. We do not put any \emph{a priori} bounds on the dimensions or
differential order of $e$. Hence, $e$ can also be seen as a matrix of
scalar differential operators. Hence, when $e$ is of order zero, $e[u]$
will correspond to multiplication of $u$ by an $r$-dependent matrix. We
will denote the composition of differential operators by $\circ$, so
that $e\circ f[u] = e[f[u]]$. We will restrict our attention only to
differential operators with rational coefficients and in general complex
valued.

In this section, we will eventually show how to find all the rational
solutions $u = u(r)$ of a rational ODE $e[u] = v$, by reducing it to a
finite dimensional linear algebra problem. Our approach starts with a
Laurent series representation $u = \sum_n u_n r^n$ converts the equation
$e[u] = v$ into a recurrence relation on the coefficients of $u$. At
different stages, it would be useful to consider Laurent series of
different kinds. In particular, we will mostly deal with formal series
(no requirement of convergence). However, convergence will be automatic
if we know in advance that the series has only finitely many terms or
that it comes from the expansion of a rational function. Thus, we may
distinguish \emph{unbounded} Laurent series $\C[[r,r^{-1}]]$,
\emph{bounded (from below)} Laurent series $\C[r^{-1}][[r]]$, \emph{bounded from
above} Laurent series $\C[r][[r^{-1}]]$ and \emph{Laurent polynomials}
$\C[r,r^{-1}]$. Of course, we could also consider Laurent series
centered at some other $r = \rho \ne 0$, but for convenience of notation
whenever possible we will stick with $\rho=0$.

For bounded (from below) Laurent series, it is helpful to define
leading or trailing orders and coefficients. For $u = \sum_n u_n r^n \in
\C[r^{-1}][[r]]$, if we can write
\begin{equation}
	u = \begin{bmatrix}
		u^1_{n_1} r^{n_1} (1 + O(r)) \\
		u^2_{n_2} r^{n_2} (1 + O(r)) \\
		\vdots
	\end{bmatrix} ,
\end{equation}
with each $O(r) \in r\C[[r]]$, the $n_1,n_2,\ldots$ are the
\emph{leading orders} of the components $u^1,u^2,\ldots$ of $u$, with
the exception when $u^i_{n_i} = 0$, in which case we set $n_i = +\oo$.
We denote by $|\check{u}|$ the vector where each component of $u$ is
replaced by its leading order, and we refer to it as the \emph{leading
order} of $u$. When $n = \min_i n_i < \oo$, $u_n \ne 0$ and we call it
the \emph{leading coefficient} of $u$. We define the leading coefficient
of $0$ to be $0$.

Similarly, for bounded from above Laurent series $u \in C[r][[r^{-1}]]$,
if we can write
\begin{equation}
	u = \begin{bmatrix}
		u^1_{n_1} r^{n_1} (1 + O(r^{-1})) \\
		u^2_{n_2} r^{n_2} (1 + O(r^{-1})) \\
		\vdots
	\end{bmatrix} ,
\end{equation}
with each $O(r^{-1}) \in r^{-1} \C[[r^{-1}]]$, the $n_1,n_2,\ldots$ are
the \emph{trailing orders} of the components $u^1,u^2,\ldots$ of $u$,
with the exception when $u^i_{n_i} = 0$, in which case we set $n_i =
-\oo$. The \emph{trailing order} $|\hat{u}|$ of $u$ is the vector of the
trailing orders of the components of $u$. When $n = \max_i n_i > -\oo$,
$u_n \ne 0$ and we call it the \emph{trailing coefficient} of $u$. We
define the trailing coefficient of $0$ to be $0$.

Clearly the leading (trailing) coefficient of a bounded (from above)
Laurent series vanishes if and only if the whole series vanishes. For
Laurent polynomials $u \in \C[r,r^{-1}]$, both the leading and trailing
orders, and coefficients, are well-defined.

Consider an ODE $e[u] = 0$ on bounded Laurent series $u = \sum_n u_n r^n
\in \C[r^{-1}][[r]]$. What we like to do is turn $e[u] = 0$ into a
linear recurrence relation on the coefficients $u_n$ of the form $E_n
u_n = f_n(u_{n-1},u_{n-2},\ldots)$ and then uniquely solve for $u_n$ as
a function of $u_{n-1}$ and lower order coefficients, for almost all $n$
(that is, all but finitely many). Those finitely many $n$ for which the
solution for $u_n$ would not be unique would then determine the
dimension of the solution space of the ODE. This approach requires that
the coefficients $E_n$ in the recurrence relation be invertible for
almost all $n$. For scalar equations this is an almost trivial
requirement, but in matrix equations different components of $e$ may be
weighted so differently by powers of $r$ that the coefficient $E_n$
comes out as a singular matrix for infinitely many $n$. Often this
problem can be remedied by applying suitable transformations to $u$ and
to $e[u]$.

Let $S = S(r)$ and $T = T(r)$ be matrices with Laurent polynomial
components. For future convenience, we also require that the inverses
$S^{-1}$ and $T^{-1}$ also have Laurent polynomial components (or,
equivalently, the determinants of $S$ and $T$ are proportional to single
powers of $r$). We say that $S$ and $T$ are respectively the
\emph{source} and \emph{target leading multipliers} of $e$ when, after
expanding all rational coefficients as bounded Laurent series, we have
\begin{equation} \label{eq:lead-mult}
	e[S(r) u_n r^n] = T(r) (E_n u_n r^n + r^n O(r)) ,
\end{equation}
with the components of $O(r)$ all in $r\C[[r]]$ and $E_n$ an
$r$-independent matrix that is invertible for almost all $n$. We call
$E_n$ the \emph{trailing characteristic matrix} of $e$ with respect to
the given multipliers. Similarly, we say that $S$ and $T$ are
respectively the \emph{source} and \emph{target trailing multipliers} of
$e$ when, after expanding all rational coefficients as bounded from
above Laurent series, we have
\begin{equation} \label{eq:trail-mult}
	e[S(r) u_n r^n] = T(r) (E_n u_n r^n + r^n O(r^{-1})) ,
\end{equation}
with the components of $O(r^{-1})$ all in $r^{-1}\C[[r^{-1}]]$ and $E_n$
an $r$-independent matrix that is invertible for almost all $n$. We call
$E_n$ the \emph{trailing characteristic matrix} of $e$ with respect to
the given multipliers.

Those integer $n\in \mathbb{Z}$ such that $\det E_n = 0$, which is a
polynomial in $n$, are called (respectively \emph{leading} or
\emph{trailing}) \emph{(integer) characteristic roots} or
\emph{exponents} of $e$ with respect to given multipliers $S,T$. We
denote the set of such leading characteristic exponents by
$\check{\sigma}(e)$ and the set of such trailing characteristic
exponents by $\hat{\sigma}(e)$, with implicit dependence on the $S,T$
multipliers, of course.

We will not dwell on when leading or trailing multipliers exist, but
will just assume that they are given for any particular problem. Often
$S$ and $T$ may be taken to be diagonal, with appropriately chosen
powers of $r$ on the diagonal. Otherwise, they could be determined by a
recursive procedure similar to that used in the analysis of regular and
irregular singularities for ODEs with meromorphic
coefficients~\cite{wasow}.

%

Any rational $u \in \C(r)$ will have a (convergent) bounded Laurent
series expansion about any point $r=\rho$. Without loss of generality,
let us take $\rho = 0$. We would like to
prove some bounds on the leading order of $u$ at $r=0$ when it solves
$e[u] = v$, with some rational $v \in \C(r)$. For this purpose, it is
actually more natural to allow $u$ and $v$ to be bounded Laurent series.

\begin{lem} \label{lem:leading-bounds}
Let $e[u]=0$ be an ODE with rational coefficients, leading multipliers
$S,T$, and leading characteristic matrix $E_n$. Let $u,v \in
\C[r^{-1}][[r]]$ with leading orders $m = \min_i |\check{u}^i|$, $n =
\min_i |\check{v}^i|$ (the values $m=\oo$ or $n=\oo$ are both
permissible). If $e[Su] = Tv$, then either (a) $m = n$ and (provided
$n<\oo$) $v_n = E_n u_n$ or (b) $m$ is a leading characteristic exponent
of $e$, $E_m u_m = 0$, and $m < n$. In other words
\begin{equation} \label{eq:leading-bounds}
	\min \; \{ n\} \cup \check{\sigma}(e) \le m .
\end{equation}
\end{lem}
This result and proof are analogous to those presented in~\textsection6
of~\cite{abramov}, where only the case of scalar equations and
polynomial coefficients is treated. The monograph~\cite{abramov}
cites~\cite{abramov89a,abramov89b,abp95} as the original sources for the
basic ingredients of the approach. A generalization of the approach to
systems of arbitrary size and order can be found in~\cite{abramov14}
(which also cites slightly earlier related work). Our innovation is to
synthesize this approach into an economic method, as presented in this
section, based on the convenient notion of leading (trailing)
multipliers $S,T$ and the way they lead to the leading (trailing)
characteristic matrix $E_n$.
\begin{proof}
If $m=\oo$, this means that $u=0$. Then also $v=0$ and $n=\oo$, meaning
that (a) holds. For the rest we will assume that $m<\oo$, meaning that
$u$ has the non-vanishing leading coefficient $u_m \ne 0$. If $E_m u_m
\ne 0$, then the defining property~\eqref{eq:lead-mult} of leading
multipliers $S$ and $T$ directly implies part (a), that is $n = m$. On
the other hand, if $E_m u_m = 0$ and since by definition $u_m \ne 0$,
the leading order of $u$ must be a characteristic exponent of $e$, $m
\in \check{\sigma}(e)$. Using again~\eqref{eq:lead-mult}, we also find
$m < n$.

Part (a) implies $n \le m$, while part (b) implies
$\min\check{\sigma}(e) \le m$ and $m<n$. Since at least one of (a) or
(b) always holds, the lower bound~\eqref{eq:leading-bounds} on $m$ is
always true.
\end{proof}

All the same arguments apply to Laurent expansions about $r=\oo$, though
after making use of the transformation $r \mapsto 1/r$. For convenience,
we state the corresponding result without the need to invoke this
transformation.

\begin{lem} \label{lem:trailing-bounds}
Let $e[u]=0$ be an ODE with rational coefficients, trailing multipliers
$S,T$, and trailing characteristic matrix $E_n$. Let $u,v \in
\C[r][[r^{-1}]]$ with trailing orders $m = \max_i |\hat{u}^i|$, $m =
\max_i |\hat{v}^i|$ (the values $m=-\oo$ or $n=-\oo$ are both
permissible). If $e[Su] = Tv$, then either (a) $m = n$ and (provided
$n>-\oo$) $v_n = E_n u_n$ or (b) $m$ is a trailing characteristic
exponent of $e$ and $E_m u_m = 0$. In other words
\begin{equation} \label{eq:trailing-bounds}
	m \le \max \; \{ n\} \cup \hat{\sigma}(e) .
\end{equation}
\end{lem}

Now we know how to bound the order of the pole of a rational solution
$u$ at any particular value of $r=\rho \in \mathbb{C}$. For the
following class of ODEs, we can also identify all the potential
locations of the poles of $u$.

\begin{lem} \label{lem:pole-locations}
Let $e[u] = 0$ be an ODE of differential order $p$ with rational
coefficients, for which there exists an invertible matrix $P=P(r)$ with
rational coefficients such that $P e[u] = \frac{d^p}{dr^p} u +
\tilde{e}[u]$, where $\tilde{e}$ is of differential order at most $p-1$.
For rational $u,v \in \C(r)$, if $e[u] = v$, then $u(r)$ is smooth
(i.e.,\ has no pole) at all but finitely many points of $\C$. The only
possible exceptions are $r=\rho$, with $\rho$ one of the poles of $Pv$
or of the coefficients of $\tilde{e}$.
\end{lem}
\begin{proof}
By our hypotheses, we can put the equation $e[u] = v$ into the
equivalent form
\begin{equation} \label{eq:standard-ode}
	\frac{d^p}{dr^p} u + \tilde{e}[u] = Pv ,
\end{equation}
where $Pv$ is rational and $\tilde{e}[u]$ has rational coefficients.
Consider a point $r=\rho \in \mathbb{C}$, that is not pole of $Pv$ or of
the coefficients of $\tilde{e}[u]$. There are obviously only finitely
many such excluded points. If $u$ has a pole of type $(r-\rho)^{-k}$,
then $\frac{d^p}{dr^p} u$ has a pole of type $(r-\rho)^{-k-p}$, while
$Pv$ and $\tilde{e}[u]$ will only have poles of lower order. But this
means that such a $u$ cannot be a solution of our equation. Hence, any
rational solution $u \in \C(r)$ can have poles only in the already
mentioned excluded set.
\end{proof}

Given a rational ODE $e[u] = 0$, when considering Laurent expansions at
$r=\rho$, let us denote the corresponding leading multipliers by
$S_\rho, T_\rho$, which are by definition rational and have poles only
at $r=\rho$ (and $r=\oo$, of course). For a rational $u \in \C(r)$, if
we know that its poles are restricted to a finite set of points in $\C$
and we have a bound on the degree of the pole at each of these points,
then we can find a rational matrix $R=R(r)$ such that $Ru$ has no poles
in $\C$.

\begin{thm} \label{thm:univ-mult}
Let $e[u] = v$ be an ODE with rational coefficients and rational $v \in
\C(r)$, satisfying the hypotheses of Lemma~\ref{lem:pole-locations}.
Suppose also that we have the leading multipliers $S_\rho,T_\rho$ of $e$
at each of the finitely many exceptional points $r=\rho\in \C$
identified in Lemma~\ref{lem:pole-locations}. Then, there exists a
rational matrix $R=R(r)$ such that, for any rational $u\in \C(r)$
satisfying $e[u] = v$, there is a Laurent polynomial $\tilde{u} \in
\C[r,r^{-1}]$ satisfying $u = R\tilde{u}$.
\end{thm}
We call such a matrix $R$ a \emph{universal multiplier} for the rational
inhomogeneous ODE $e[u] = v$. A universal multiplier certainly need not
be unique. The existence of universal multipliers for scalar equations
is discussed in~\cite[\textsection7]{abramov}, which
cites~\cite{abramov89b,abp95} as original references. For systems of
arbitrary size and order, the existence of universal multipliers is
discussed for instance in~\cite{abramov14} (with references to slightly
earlier work).
\begin{proof}
A rational $u \in \C(r)$ has only finitely many poles, and at each of
those poles it has a bounded Laurent series expansion. By invoking
Lemma~\ref{lem:pole-locations} we can constrain the poles of $u$ to a
finite set of points. Then, by invoking Lemma~\ref{lem:leading-bounds},
for each of those points, say $r=\rho$, we can find a finite lower bound
$\check{n}_\rho$ for the leading Laurent order of $S^{-1}_\rho u$ at
$r=\rho$. Recall that one of the defining properties of $S_\rho$ is that
both it and $S^{-1}_\rho$ only have poles at $r=\rho$ (and of course at
$r=\oo$). This means that $\tilde{u} = \prod_\rho S^{-1}_\rho
(r-\rho)^{-\check{n}_\rho} u$, where the product is taken over the
potential pole locations (possibly excluding $\rho=0$), is still
rational but no longer has any poles in $r\in \C$, with the possible
exception of $r=0$. But that can only be if $\tilde{u} \in \C[r,r^{-1}]$
is a Laurent polynomial. Therefore, we can take
\begin{equation} \label{eq:univ-mult}
	R(r) = \prod_{\rho} S_\rho(r) (r-\rho)^{\check{n}_\rho}
\end{equation}
as the desired universal multiplier. Since any of the $\check{n}_\rho$
can be decreased without breaking this result, we have many possible
choices for $R$.
\end{proof}

\begin{cor}
Let $e$ and $v$ be as in Theorem~\ref{thm:univ-mult}, with universal
multiplier $R$. In addition, suppose that we have the leading
multipliers $S_0,T_0$ at $r=0$ and the trailing multipliers
$S_\oo,T_\oo$ at $r=\oo$ for $\tilde{e} = e \circ R$. Then the equation
$e[u] = v$ for $u$ can be reduced to a finite dimensional linear system,
and hence its solution space is finite dimensional.
\end{cor}
\begin{proof}
By invoking Theorem~\ref{thm:univ-mult}, solving $e[u] = v$ for $u\in
\C(r)$ is equivalent to solving $\tilde{e}[\tilde{u}] = v$ for
$\tilde{u} \in \C[r,r^{-1}]$, with $u = R \tilde{u}$ and
$\tilde{e}[\tilde{u}] = e[Ru]$. Invoking Lemma~\ref{lem:leading-bounds}
we can find a finite lower bound on the leading order of $S^{-1}_0
\tilde{u}$ and hence of $\tilde{u}$, which we will call $\check{n}$.
Invoking Lemma~\ref{lem:trailing-bounds} we can find a finite upper
bound on the leading order of $S^{-1}_\oo \tilde{u}$ and hence of
$\tilde{u}$, which we will call $\hat{n}$. Therefore, we can parametrize
all solutions as Laurent polynomials
\begin{equation} \label{eq:u-ansatz}
	\tilde{u} = \sum_{n=\check{n}}^{\hat{n}} u_n r^n ,
\end{equation}
which has $\hat{n} - \check{n} + 1 < \oo$ undetermined coefficients.
Plugging this parametrization into the equation $\tilde{e}[\tilde{u}] =
v$, putting both sides over a common denominator, and comparing
coefficients reduces the original problem to a finite dimensional linear
system of equations. The dimension of the solution space of this system
is of course finite, and (crudely) bounded by the number of coefficients
in~\eqref{eq:u-ansatz}.
\end{proof}

Of course, once an equation has been reduced to an explicit finite
dimensional linear system, it can be solved on a computer, even
symbolically.


\section{Reducing triangular ODE systems with rational coefficients}
\label{sec:offdiag}

In this section, we are interested in the following question. Given an
ODE system in block upper triangular form, is it possible to find an
equivalent ODE system where the off-diagonal block has been set to zero,
hence in diagonal form? If possible, we call this a \emph{reduction to
block diagonal form} and say that the original system can be
\emph{reduced}. A refined version of the question is whether a rational
ODE system can be reduced while remaining rational.

The first thing we need to clarify is the notion of equivalence. Roughly
speaking, two ODE systems should be equivalent when there is an
isomorphism between their solution spaces. A further practical
requirement is that this isomorphism be given, in either direction, by
differential operators. After all, transformations given by differential
operators tend to be easier to write down in terms of explicit formulas,
while also allowing rather precise control over some properties of the
coefficients of the ODE systems, like rationality or upper triangular
form. Also, an equivalence should make explicit the transformation of
one ODE system into the other one, again by a differential operator.

We formalize these ideas as follows. Given two ODE systems, $e[u] = 0$
and $\bar{e}[\bar{u}] = 0$, an equivalence between them consists of
pairs of differential operators $k,g$ and $\bar{k},\bar{g}$ obeying the
operator identities, for any $u,\bar{u}$ and $v,\bar{v}$,
\begin{align}
	\bar{e}[k[u]] &= g[e[u]] ,
		&
	\bar{k}[k[u]] &= u ,
		&
	\bar{g}[g[v]] &= v ,
	\\
	e[\bar{k}[\bar{u}]] &= \bar{g}[\bar{e}[\bar{u}]] ,
		&
	k[\bar{k}[\bar{u}]] &= \bar{u} ,
		&
	g[\bar{g}[\bar{v}]] &= \bar{v} .
\end{align}
Graphically, if we represent each differential operator by an arrow and
appropriate function spaces by $\bullet$'s, these identities mean that
the squares in the following diagram are commutative and that the
horizontal arrows compose to identity in either direction:
\begin{equation}
\begin{tikzcd}[column sep=large,row sep=large]
	\bullet \ar{d}{e} \ar[shift left]{r}{k} \&
	\bullet \ar[swap]{d}{\bar e} \ar[shift left]{l}{\bar k}
	\\
	\bullet \ar[shift left]{r}{g} \&
	\bullet \ar[shift left]{l}{\bar g}
\end{tikzcd} \, .
\end{equation}
Basically, these identities imply that for a solution $u$ of $e[u] = 0$,
$\bar{u} = k[u]$ is a solution of $\bar{e}[\bar{u}] = 0$, and vice
versa, where the barred and unbarred transformation operators are
mutually inverse. Finally, when dealing with ODE systems with rational
coefficients, we require the coefficients of the operators $k,g$ and
$\bar{k},\bar{g}$ to be rational as well.

The above notion of equivalence is actually somewhat more rigid than
absolutely necessary, but it will be sufficient for our purposes. A
discussion of a somewhat looser notion of equivalence can be found
in~\cite{kh-vwtriang}, with references to deeper literature on this
topic. Below, we use this notion of equivalence to discuss reduction of
triangular ODE systems. A similar discussion can already be found
in~\cite[Sec.2.3]{kh-vwtriang}.

An ODE system of the form
\begin{equation} \label{eq:triang-ode}
	\begin{bmatrix}
		e_0 & \Delta \\
		0 & e_1
	\end{bmatrix}
	\begin{bmatrix}
		u_0 \\ u_1
	\end{bmatrix}
	=
	\begin{bmatrix}
		0 \\ 0
	\end{bmatrix}
\end{equation}
is said to be \emph{(block) upper triangular}, or \emph{(block)
diagonal} if $\Delta = 0$. We will always assume that this system has
rational coefficients. We presume also that the equations $e_0[u_0]=0$
and $e_1[u_1]=0$ are ODE systems of unspecified dimensions and
differential orders, but such that they can be solved for the highest
derivatives, as in the hypotheses of Lemma~\ref{lem:pole-locations}.

A reduction of the upper triangular ODE system~\eqref{eq:triang-ode} is
an equivalence given by the following pair of commutative diagrams
\begin{equation} \label{eq:triang-reduce}
\begin{tikzcd}[column sep=large,row sep=large]
	\bullet
		\ar[swap]{d}{\begin{bmatrix} e_0 & \Delta \\ 0 & e_1 \end{bmatrix}}
		\ar{r}{\begin{bmatrix} \id & \delta \\ 0 & \id \end{bmatrix}}
	\&
	\bullet
		\ar{d}{\begin{bmatrix} e_0 & 0 \\ 0 & e_1 \end{bmatrix}}
	\\
	\bullet
		\ar[swap]{r}{\begin{bmatrix} \id & \eps \\ 0 & \id \end{bmatrix}}
	\&
	\bullet
\end{tikzcd}
	\, , \quad
\begin{tikzcd}[column sep=large,row sep=large]
	\bullet
		\ar[swap]{d}{\begin{bmatrix} e_0 & 0 \\ 0 & e_1 \end{bmatrix}}
		\ar{r}{\begin{bmatrix} \id & -\delta \\ 0 & \id \end{bmatrix}}
	\&
	\bullet
		\ar{d}{\begin{bmatrix} e_0 & \Delta \\ 0 & e_1 \end{bmatrix}}
	\\
	\bullet
		\ar[swap]{r}{\begin{bmatrix} \id & -\eps \\ 0 & \id \end{bmatrix}}
	\&
	\bullet
\end{tikzcd} \, ,
\end{equation}
where the corresponding horizontal arrows are clearly mutual inverses.
Of course, we require the differential operators $\delta$ and $\eps$ to
have rational coefficients. By direct calculation, we can check that the
above diagrams commute if and only if $\delta$ and $\eps$
satisfy the operator identity
\begin{equation} \label{eq:offdiag}
	e_0 \circ \delta = \Delta + \eps \circ e_1 .
\end{equation}

Note that solutions of~\eqref{eq:offdiag} are certainly not unique. For
instance, for any $\delta,\eps$ solution pair, $(\delta + \alpha\circ
e_1), (\eps + e_0\circ \alpha)$ is another solution, with arbitrary
$\alpha$, since $e_0 \circ (\alpha \circ e_1) = (e_0\circ \alpha) \circ
e_1$. In addition, having a solution pair $\delta,\eps$ for a given
$\Delta$, automatically gives us the solution pairs $(\delta - \alpha),
(\eps+\beta)$ for $\Delta$ replaced with $\Delta + e_0 \circ \alpha +
\beta\circ e_1$. When $e_0[u_0] = 0$ and $e_1[u_1] = 0$ can be solved
for their highest derivatives, we can use the above freedom to reduce
equation~\eqref{eq:offdiag}, with $\delta,\eps$ and $\Delta$ of
potentially high differential orders, to the same equation, but with the
differential orders of $\delta,\eps$ and $\Delta$ bounded by the orders
of $e_0$ and $e_1$.

\begin{thm} \label{thm:reduce-order}
Suppose that the rational ODE systems $e_0[u_0] = 0$ and $e_1[u_1] = 0$
of differential orders $p_0$ and $p_1$, respectively. Suppose also that
they can be solved for the highest order derivatives, that is, for
$i=0,1$ there exist rational invertible matrices $P_i$ such that $P_i
e_i[u] = \frac{d^{p_i}}{dr^{p_i}} u + \tilde{e}_i[u]$, where
$\tilde{e}_i$ is of differential order $<p_i$. (a) The
knowledge of $\delta$ and $\Delta$ in~\eqref{eq:offdiag} is sufficient
to reconstruct $\eps$ uniquely. (b) For given $\Delta$ and a $\delta$ of
fixed differential order, the existence of an $\eps$
satisfying~\eqref{eq:offdiag} is equivalent to a rational ODE system on
the coefficients of $\delta$. (c) If $\Delta$ if of differential order
$<p_0+p_1$, then~\eqref{eq:offdiag} has a solution if and only if it has
a solution where $\delta$ is of differential order $<p_1$ and $\eps$ is
of differential order $<p_0$.
\end{thm}
\begin{proof}
We first make the standard observation is that, under our hypotheses on
$e_i$ ($i=0,1$), for any differential operator $f_i[u_i]$, we can find a
unique differential operators $g_i$ and $\tilde{f}_i$ such that $f_i =
g_i + \tilde{f}_i \circ e_i$, with $g_i$ of differential order $<p_i$.
This is easy to prove by noting that we cannot decrease the differential
order of $e_i$ by pre-composing it with a non-zero differential operator
and then recursively rewriting the highest order derivatives in $f_i$,
say $\frac{d^{p_i+q}}{dr^{p_i+q}} u_i$, as $-\frac{d^q}{dr^q}
\tilde{e}_i[u] + \frac{d^q}{dr^q} P_i e_i[u_i]$. Obviously, both $g_i$
and $\tilde{f}_i$ also have rational coefficients and, in fact, their
coefficients are linear rational differential operators applied to the
coefficients of $f_i$.

To prove part (a), note that the identity $e_0\circ \delta - \Delta = 0
+ \eps\circ e_1$, combined with our initial observation, implies that
$\eps$ is uniquely fixed once we know $\Delta$ and $\delta$.

To prove part (b), consider the decomposition $e_0\circ \delta - \Delta
= \tilde{\Delta} + \eps\circ e_1$, with $\tilde{\Delta}$ of differential
order $<p_1$, which by our initial observation always exists and is
unique. Thus, $\delta,\eps$ and $\Delta$ satisfy~\eqref{eq:offdiag} if
and only if the coefficients of $\tilde{\Delta}$ are all zero. But
construction, the coefficients of $\tilde{\Delta}$ are linear rational
differential operators acting on the coefficients of $\delta$ and
$\Delta$.

To prove part (c), we first apply our initial observation to get the
decomposition $\delta = \tilde{\delta} + \eps_1\circ e_1$, where the
differential order of $\tilde{\delta}$ is $<p_1$. Then $e_0 \circ
\tilde{\delta} = \Delta + \tilde{\eps}\circ e_1$, with $\tilde{\eps} =
\eps-\eps_1$. The differential orders of $e_0 \circ \tilde{\delta}$ and
$\Delta$ are both $<p_0+p_1$, hence by comparison we can conclude that
the differential order of $\tilde{\eps}$ is $<p_0$.
\end{proof}

\section{Systems of Regge-Wheeler equations}
\label{sec:rw}

Let
\begin{equation}
	f(r) = 1 - \frac{2M}{r}, \quad
	f'(r) = \frac{f_1(r)}{r}, \quad
	f_1(r) = \frac{2M}{r} .
\end{equation}
Define the \emph{(generalized) spin-$s$ Regge-Wheeler operator} with
\emph{mass parameter} $M$, \emph{angular momentum quantum number} $l$
and \emph{frequency} $\omega$ by
\begin{equation}
	\D_s \phi = \del_r f \del_r \phi
		- \frac{1}{r^2} [\B_l + (1-s^2)f_1] \phi + \frac{\omega^2}{f} \phi ,
\end{equation}
where $\B_l = l(l+1)$, with $l=0,1,2,\ldots$. We will assume that
$\omega \ne 0$ and that $s$ is a non-negative integer. Of course, for
any $s$, $\D_s$ has rational coefficients and satisfies the hypotheses
of Lemma~\ref{lem:pole-locations}.

Consider the upper triangular rational ODE system
\begin{equation} \label{eq:rw-sys}
	\begin{bmatrix}
		\D_{s_0} & \Delta \\
		0 & \D_{s_1}
	\end{bmatrix}
	\begin{bmatrix}
		u_0 \\ u_1
	\end{bmatrix}
	=
	\begin{bmatrix}
		0 \\ 0
	\end{bmatrix} ,
\end{equation}
where we suppose that $\Delta$ is of differential order at most $1$. As
discussed in Section~\ref{sec:offdiag}, this system is reducible to
diagonal by an equivalence (Section~\ref{sec:offdiag}) with rational
coefficients if and only if the following version of
Equation~\ref{eq:offdiag} is satisfied:
\begin{equation} \label{eq:rw-offdiag}
	\D_{s_0} \circ \delta = \Delta + \eps \circ \D_{s_1} .
\end{equation}
By Theorem~\ref{thm:reduce-order}, without loss of generality, we can
consider this problem restricted to the following class of
operators:
\begin{align}
\label{eq:Delta-param}
	\Delta &= \frac{1}{r^2} (\Delta_1 r \del_r + \Delta_0) , \\
\label{eq:delta-param}
	\delta &= \delta_1 r \del_r + \delta_0 , \\
\label{eq:eps-param}
	\eps &= \delta_1 r \del_r + [2\del_r(r\delta_1)-\frac{f_1}{f}\delta_1
		+ \delta_0] ,
\end{align}
where $\delta_i,\Delta_i$, for $i=0,1$, are all rational functions.
Plugging this parametrization into~\eqref{eq:rw-offdiag} and comparing
coefficients, we find the equivalent ODE system
\begin{multline} \label{eq:rw-decoupling}
	e \begin{bmatrix}
		\delta_0 \\ \delta_1
	\end{bmatrix}
	:=
	\begin{bmatrix}
		f & 0 \\
		0 & f
	\end{bmatrix}
	r^2 \del_r^2 \begin{bmatrix} \delta_0 \\ \delta_1 \end{bmatrix}
	+
	\begin{bmatrix}
		f_1 & -2\frac{\omega^2 r^2}{f} + 2[\B_l+f_1(1-s_1^2)] \\
		2 f & 2f-f_1
	\end{bmatrix}
	r \del_r \begin{bmatrix} \delta_0 \\ \delta_1 \end{bmatrix}
	\\
	+
	\begin{bmatrix}
		f_1 (s_0^2-s_1^2) & -2\frac{\omega^2 r^2(f-f_1)}{f^2} - \frac{f_1}{f} [\B_l+1-s_1^2] \\
		0 & f_1 (s_0^2-s_1^2 + \frac{1}{f})
	\end{bmatrix}
	\begin{bmatrix} \delta_0 \\ \delta_1 \end{bmatrix}
	=
	\begin{bmatrix} \Delta_0 \\ \Delta_1 \end{bmatrix}
\end{multline}
for $\delta_0,\delta_1$, with $\Delta_0,\Delta_1$ as inhomogeneous
sources.

Next, we will apply the analysis of Section~\ref{sec:rat-sols} to check
the conditions under which the system~\eqref{eq:rw-decoupling} has
rational solutions for $\delta_0,\delta_1$, with given
$\Delta_0,\Delta_1$.

It is easy to see that the only singular points of the
equation~\eqref{eq:rw-decoupling} are $r=0,2M,\oo$ (recall that $f(2M) =
0$). In this work, we will not consider $\Delta_0, \Delta_1$ with poles
at other values of $r$, which means that these points are the only
possible locations of the poles of $\delta_0, \delta_1$.


For each of the singular points, we have the following multipliers ($S$
and $T$), characteristic matrices ($E_n$) and characteristic exponents
($\sigma(e)$).
\begin{itemize}
\item $r=0$:
	$\check{\sigma}_0(e) = \{\pm s_0 \pm s_1\}$, with
	$\det E_n = (2M)^2 (n+s_0+s_1) (n+s_0-s_1) (n-s_0+s_1) (n-s_0-s_1)$
	and
	\begin{equation}
		S = \begin{bmatrix}
			1 & 0 \\
			0 & 1
		\end{bmatrix} , \quad
		T = \begin{bmatrix}
			r^{-1} & 0 \\
			0 & r^{-1}
		\end{bmatrix} , \quad
		E_n = -2M \begin{bmatrix}
			n^2 - 2n - s_0^2 + s_1^2 & -2n (1-s_1^2) \\
			2n & n^2 + 2n - s_0^2 + s_1^2
		\end{bmatrix} .
	\end{equation}
\item $r=2M$:
	$\check{\sigma}_{2M}(e) = \{ -1 \}$, with
	$\det E_n = (n+1)^2 [(n+1)^2 + 16M^2\omega^2]$
	and
	\begin{equation}
		S = \begin{bmatrix}
			\frac{(r-2M)}{2M} & 0 \\
			0 & \frac{(r-2M)^2}{4M^2}
		\end{bmatrix} , \quad
		T = \begin{bmatrix}
			1 & 0 \\
			0 & \frac{2M}{(r-2M)}
		\end{bmatrix} , \quad
		E_n = \begin{bmatrix}
			(n+1)^2 & -8M^2\omega^2 (n+1) \\
			2(n+1) & (n+1)^2
		\end{bmatrix} .
	\end{equation}
\item $r=\oo$:
	$\hat{\sigma}_\oo(e) = \{ -1, 0 \}$, with
	$\det E_n = 4\omega^2 n (n+1)$
	and
	\begin{equation}
		S = \begin{bmatrix}
			1 & 0 \\
			0 & 1
		\end{bmatrix} , \quad
		T = \begin{bmatrix}
			r^2 & 0 \\
			0 & 1
		\end{bmatrix} , \quad
		E_n = \begin{bmatrix}
			0 & -2\omega^2 (n+1) \\
			2n & n(n+1)
		\end{bmatrix} .
	\end{equation}
\end{itemize}

Suppose that $A_\rho = T^{-1}_\rho [\begin{smallmatrix} \Delta_0 \\
\Delta_1 \end{smallmatrix}]$ has leading order $\check{m} = \min_i
|\check{A}^i_0|$, and trailing order $\hat{m} = \max_i |\hat{A}^i_\oo|$.
At $r=2M$, we do not need the specific order, but just some lower bound
$m \le \min_i |\check{A}^i_{2M}|_{r=2M}$ for some integer $m\le -1$. We
choose a bound of this form because then the identity
\begin{equation}
	\min\{m\} \cup \check{\sigma}_{2M}(e)
	= \min \{m, -1\} = m
\end{equation}
determines that the Laurent series expansion of the solution $\delta =
[\begin{smallmatrix} \delta_0 \\ \delta_1 \end{smallmatrix}]$ must
belong to
\begin{equation}
	\delta
	\in \begin{bmatrix}
			\frac{(r-2M)}{2M} & 0 \\
			0 & \frac{(r-2M)^2}{4M^2}
		\end{bmatrix}
		(r-2M)^{\min\{m\} \cup \check{\sigma}_{2M}(e)} \C[[(r-2M)]]
	= \begin{bmatrix}
			1 & 0 \\
			0 & \frac{(r-2M)}{2M}
		\end{bmatrix}
		(r-2M)^{m+1} \C[[r]] .
\end{equation}
Since this is the only condition to be satisfied for poles other than at
$r=0,\oo$, without loss of generality, we can take $R=f^{m+1} =
(\frac{r-2M}{2M})^{m+1}$ as a convenient universal multiplier. So that,
according to Theorem~\ref{thm:univ-mult}, any rational solution
of~\eqref{eq:rw-decoupling} must satisfy
\begin{equation}
	\delta \in R \, \C[r,r^{-1}] = f^{m+1} \C[r,r^{-1}] .
\end{equation}
Finally, we can parametrize any such solution with the following bounded
order Laurent polynomial:
\begin{equation} \label{eq:bounded-rat-sol}
	\delta
	= f^{m+1} \sum_{n=\check{n}}^{\hat{n}} d_n r^n ,
	\quad \text{where} \quad
	\left\{
	\begin{aligned}
		\hat{n} &= \max \{\hat{m}\} \cup \sigma_\oo(e)
			= \max\{\hat{m},0\} , \\
		\check{n} &= \min \{\check{m}\} \cup \sigma_0(e)
			= \min\{\check{m},-s_0-s_1\} .
	\end{aligned}
	\right.
\end{equation}

\subsection{Examples} \label{sec:examples}

We finish with a few explicit examples. Sometimes, in specific examples,
a $\delta,\epsilon$ solution for a given $\Delta$ can be found by trial
and error. However, when unguided, such a process can be quite
laborious. And at the end, if no solution was found, one cannot
automatically conclude that a solution does not exist. Using the method
presented above, when the trial and error method becomes too laborious,
it can be automated using a computer algebra system. Moreover, our
method can also furnish a proof that in some situation no solution
exists.

Below, we find it helpful to use the notation $\check{O}(r^p)$ to denote
the leading order of a Laurent series at $r=0$ and $\hat{O}(r^q)$ to
denote the trailing order of a Laurent series at $r=\oo$.

\begin{enumerate}
\item \label{itm:f1-s01}
	The equation
	\begin{equation}
		\D_0 \circ \delta = \frac{f_1}{r^2} + \eps \circ \D_1 ,
	\end{equation}
	where
	\begin{equation}
		\begin{bmatrix}
			\Delta_0 \\ \Delta_1
		\end{bmatrix}
		= \begin{bmatrix}
			f_1 \\ 0
		\end{bmatrix}
		= \begin{bmatrix}
			\check{O}(r^{-1}) \\ 0
		\end{bmatrix}
		= \begin{bmatrix}
			\hat{O}(r^{-1}) \\ 0
		\end{bmatrix} ,
	\end{equation}
	gives rise to the normalized sources
	\begin{equation}
		|\check{A}_0|
		= \begin{bmatrix}
				0 \\ \oo
			\end{bmatrix} , \quad
		|\check{A}_{2M}|
		 = \begin{bmatrix}
				0 \\ \oo
			\end{bmatrix} , \quad
		|\hat{A}_\oo|
		= \begin{bmatrix}
				-3 \\ -\oo
			\end{bmatrix}
		.
	\end{equation}
	Since $A_{2M} = T_{2M}^{-1}\Delta$ has no pole at $r=2M$, we can
	choose the universal multiplier to be $R = 1$. $A_0$ and $A_\oo$ give
	the orders $\check{m} = 0$ and $\hat{m} = -3$ and hence the Laurent
	polynomial bounds $\check{n} = \min\{0,-0-1\} = -1$ and $\hat{n} =
	\max\{-3, 0\}$. From this information, we can conclude that there
	exists the \emph{unique solution}
	\begin{equation}
		\delta = -1 , \quad
		\eps = -1 .
	\end{equation}
\item \label{itm:f1-s00}
	The equation
	\begin{equation}
		\D_0 \circ \delta = \frac{f_1}{r^2} + \eps \circ \D_0
	\end{equation}
	where $\Delta_0$, $\Delta_1$ and the corresponding normalized sources
	are the same as in Example~\ref{itm:f1-s00}. Again, since $A_{2M} =
	T_{2M}^{-1}\Delta$ has no pole at $r=2M$, we can choose the universal
	multiplier to be $R = 1$. $A_0$ and $A_\oo$ give the orders $\check{m}
	= 0$ and $\hat{m} = -3$ and hence the Laurent polynomial bounds
	$\check{n} = \min\{0,-0-0\} = 0$ and $\hat{n} = \max\{-3,0\} = 0$.
	From this information, we can conclude that there exists \emph{no
	solution} for $\delta,\eps$.
\item \label{itm:vw-offdiag}
	The equation
	\begin{equation}
		\D_0 \circ \delta
		= -\frac{f_1}{r^2} \left(\B_l + \frac{f_1}{2}\right)
			+ \eps \circ \D_0
	\end{equation}
	where
	\begin{equation}
		\begin{bmatrix}
			\Delta_0 \\ \Delta_1
		\end{bmatrix}
		= \begin{bmatrix}
			-f_1 \left(\B_l + \frac{f_1}{2}\right) \\
			0
		\end{bmatrix}
		= \begin{bmatrix}
			\check{O}(r^{-2}) \\ 0
		\end{bmatrix}
		= \begin{bmatrix}
			\hat{O}(r^{-1}) \\ 0
		\end{bmatrix}
	\end{equation}
	gives rise to the normalized sources
	\begin{equation}
		|\check{A}_0|
		= \begin{bmatrix}
				-1 \\ \oo
			\end{bmatrix} , \quad
		|\check{A}_{2M}|
		 = \begin{bmatrix}
				0 \\ \oo
			\end{bmatrix} , \quad
		|\hat{A}_\oo|
		= \begin{bmatrix}
				-3 \\ -\oo
			\end{bmatrix}
		.
	\end{equation}
	Since $A_{2M} = T_{2M}^{-1}\Delta$ has no pole at $r=2M$, we can
	choose the universal multiplier to be $R = 1$. $A_0$ and $A_\oo$ give
	the orders $\check{m} = -1$ and $\hat{m} = -3$ and hence the Laurent
	polynomial bounds $\check{n} = \min\{-1,-0-0\} = -1$ and $\hat{n} =
	\max\{-3,0\} = 0$. From this information, we can conclude that there
	exists \emph{no solution} for $\delta,\eps$.
\item \label{itm:berndtson-s02}
	The equation
	\begin{equation}
		\D_0 \circ \delta = \Delta + \eps \circ \D_2 ,
	\end{equation}
	with
	\begin{equation}
	\begin{split}
		\Delta =&
			24 i f_1  r^2 \omega^3 - 4 i f (6 f f_1 +
			6 \B_l f_1 +\A_l)  r\omega \del_r\\
			&- 2 i \left( \A_l + 2 (\B_l-3) + (\A_l - \B_l) (1 + 2 \B_l) + 2 (\A_l + 6 \B_l) f - 9\frac{\A_l}{\B_l}  f^2 - 
			12 f^3 \right)  \omega \\
			&+\frac{ f_1 f\B_l  (- 4 f_1^2 + 8 f f_1 - 4 \B_l + 16 f\B_l +\A_l)}{ir\omega}\del_r + \frac{i f_1 \B_l (
			\A_l (\B_l - 7 f) + 12 f (1 - (2 + \B_l) f + f^2)) }{r^2 \omega}
	\end{split}
	\end{equation}
	\begin{equation}
		\begin{bmatrix}
			\Delta_0 \\ \Delta_1
		\end{bmatrix}
		= \begin{bmatrix}
			\check{O}(r^{-4}) \\
			\check{O}(r^{-4})
		\end{bmatrix}
		= \begin{bmatrix}
			\hat{O}(r^{3}) \\
			\hat{O}(r^{2})
		\end{bmatrix}
	\end{equation}
	gives rise to the normalized sources
	\begin{equation}
		|\check{A}_0| = \begin{bmatrix}
				-3 \\ -3
			\end{bmatrix} , \quad
		|\check{A}_{2M}| = \begin{bmatrix}
				0 \\ 1
			\end{bmatrix} , \quad
		|\hat{A}_\oo| = \begin{bmatrix}
				1 \\ 2
			\end{bmatrix} .
	\end{equation}
	Since $A_{2M} = T_{2M}^{-1}\Delta$ has no pole at $r=2M$, we can
	choose the universal multiplier to be $R = 1$. $A_0$ and $A_\oo$ give
	the orders $\check{m} = -3$ and $\hat{m} = 2$ and hence the Laurent
	polynomial bounds $\check{n} = \min\{-3,-0-2\} = -3$ and $\hat{n} =
	\max\{2,0\} = 2$. From this information, we can conclude that there
	exists a unique solution
	\begin{equation}
	\begin{split}
		\delta =&
			-6 i f_1 f r^3 \omega \del_r  - i (6 f f_1 + 12 \B_l f_1   		+\A_l)  r^2 \omega 
			+\frac{i f \B_l  (4\A_l - 4 - 2 f_1 + 24 f f_1 + 3 \B_l + f l(l-1)) r\del_r }{2 \omega}\\
			& + \frac{i   (\A_l^2 + 2 \A_l (9 + 5 \B_l) f - 6 \B_l f^2 (-8f_1 -6 + 3 \B_l)) }{4 \omega}
		,
	\end{split}
	\end{equation}
	with $\eps$ given by~\eqref{eq:eps-param}.
	The above result was obtained and checked with computer algebra.
\item \label{itm:berndtson-s12}
	The equation
	\begin{equation}
		\D_1 \circ \delta = \Delta + \eps \circ \D_2 ,
	\end{equation}
	with
	\begin{equation}
		\Delta =
			-24 i f_1 f  r\omega \del_r - 4 i \A_l  \omega
			+ \frac{6 f_1 f (3 f - 1) \B_l}{ir\omega}\del_r
			- \frac{-i f_1 \B_l  (18 f f_1  - 6 f\B_l  +\A_l)}{r^2 \omega}
	\end{equation}
	\begin{equation}
		\begin{bmatrix}
			\Delta_0 \\ \Delta_1
		\end{bmatrix}
		= \begin{bmatrix}
			\check{O}(r^{-3}) \\
			\check{O}(r^{-3})
		\end{bmatrix}
		= \begin{bmatrix}
			\hat{O}(r^{2}) \\
			\hat{O}(r^{1})
		\end{bmatrix}
	\end{equation}
	gives rise to the normalized sources
	\begin{equation}
		|\check{A}_0| = \begin{bmatrix}
				-2 \\ -2
			\end{bmatrix} , \quad
		|\check{A}_{2M}| = \begin{bmatrix}
				0 \\ 1
			\end{bmatrix} , \quad
		|\hat{A}_\oo| = \begin{bmatrix}
				0 \\ 1
			\end{bmatrix} .
	\end{equation}
	Since $A_{2M} = T_{2M}^{-1}\Delta$ has no pole at $r=2M$, we can
	choose the universal multiplier to be $R = 1$. $A_0$ and $A_\oo$ give
	the orders $\check{m} = -2$ and $\hat{m} = 1$ and hence the Laurent
	polynomial bounds $\check{n} = \min\{-2,-1-2\} = -3$ and $\hat{n} =
	\max\{1,0\} = 1$. From this information, we can conclude that there
	exists a unique solution
	\begin{equation}
		\delta =
			-12 i f_1 r^2 \omega + \frac{2 i
			f (\B_l - 2 f_1) (\B_l - 2 f + f_1)  r\del_r }{\omega} + \frac{i
			(\frac{\A_l^2}{\B_l}+ 6\frac{\A_l}{\B_l}(\B_l + 3) f -18 (\B_l-4)f^2 -36f)}{3 \omega}
		,
	\end{equation}
	with $\eps$ given by~\eqref{eq:eps-param}.
	The above result was obtained and checked with computer algebra.
\end{enumerate}

The solution from Example~\ref{itm:f1-s01} can be generalized to
\begin{equation}
	\D_{s_0} \frac{1}{(s_0^2-s_1^2)}
	= \frac{f_1}{r^2} + \frac{1}{(s_0^2-s_1^2)} \D_{s_1} ,
\end{equation}
which actually works for any complex values of $s_0,s_1$, except for
$s_0 = \pm s_1$. We obtained this parametric solution by trial and
error, while trying to understand and generalize some identities
from~\cite{berndtson}. However, simply having this formula does not tell
us whether it is the unique solution. Applying our systematic approach,
we can check uniqueness as we did in Example~\ref{itm:f1-s01} for
$s_0=0$ and $s_1=1$, but only for specific values of $s_0,s_1$ at a time.
The reason is that the lower bound $\check{n}$ on the Laurent polynomial
order of $\delta$ in~\eqref{eq:bounded-rat-sol} is influenced by $\min
\sigma_0(e) = -s_0-s_1$ (at least for non-negative integer values of the
$s_i$), which depend on the $s_i$.

Note also that the above formula is singular for $s_0=\pm s_1$ and no
longer tells us anything about the existence of solutions in those
cases. On the other hand, our systematic approach can check that
indeed no solution exists, again on a case by case basis, as we did
for $s_0=0$ in Example~\ref{itm:f1-s00}.

In Equation~(88) of~\cite{kh-vwtriang}, we had managed to reduce a
$3\times 3$ upper triangular system with Regge-Wheeler operators on the
diagonal and a single rational non-vanishing off-diagonal component.
Incidentally, the solution from Example~\ref{itm:f1-s01} was
instrumental in that simplification. The existence of a solution in
Example~\ref{itm:vw-offdiag} would mean that this system could be
further reduced to diagonal form by a rational transformation. This
question was left open in~\cite{kh-vwtriang}. However, the non-existence
of such solutions, as we have just confirmed in
Example~\ref{itm:f1-s01}, proves that no such further simplification is
possible.

In~\cite{kh-vwtriang}, we have arrived at an upper-triangular
Regge-Wheeler system as a partial decoupling of the radial mode
equations of the vector wave equation (which could be interpreted as the
\emph{harmonic} or \emph{Lorenz} gauge-fixed version of Maxwell's
equations) on the background of a Schwarzschild black hole. That work
was strongly inspired by~\cite{berndtson}, which achieved a similar
decoupling for the Lichnerowicz equation (which could be interpreted as
the \emph{harmonic} or \emph{de~Donder} gauge-fixed version of
linearized Einstein's equations) also on Schwarzschild. The methods and
results result achieved in~\cite{berndtson} are unfortunately somewhat
obscure and implicit. We aim to clarify those results using the
systematic methods that we have outlined in~\cite{kh-vwtriang} and in
the present work. For instance, the existence of the solutions from
Examples~\ref{itm:berndtson-s02} and~\ref{itm:berndtson-s12} is
equivalent to Equations~(3.49--51) from~\cite{berndtson}, which were
instrumental to their main decoupling results, but apparently obtained
by trail and error, without a clear guide to how they could be
reproduced independently. Fortunately, Examples~\ref{itm:berndtson-s02}
and~\ref{itm:berndtson-s12} show that our systematic approach can
rediscover these formulas in a straight forward way using computer
algebra.

\section{Discussion} \label{sec:discuss}

The main goal of this work was to conclusively decide when it is or is
not possible to reduce an upper triangular rational ODE system
like~\eqref{eq:rw-sys} to diagonal form by a transformation
like~\eqref{eq:triang-reduce} with rational coefficients, where on the
diagonal we have generalized Regge-Wheeler operators. This question was
left open in our previous work~\cite{kh-vwtriang}. In
Section~\ref{sec:offdiag}, we showed how to reduce this question to the
existence of a rational solution to an auxiliary rational ODE system. In
Section~\ref{sec:rat-sols} we showed that, under mild hypotheses, the
existence of a rational solution of a rational ODE system can be reduced
to a finite dimensional linear algebra problem. Hence, such a question
can always be conclusively decided, at least on a case by case basis. In
Section~\ref{sec:examples}, we gave several examples illustrating our
methods. These examples reproduce, in a systematic way, some identities
previously discovered by voluminous trial and error calculations
in~\cite{berndtson}.

These identities were used in~\cite{kh-vwtriang} to significantly
simplify, after a separation of variables, the coupled radial mode
equations of the vector wave equation on Schwarzschild spacetime. Our
Example~\ref{itm:vw-offdiag} shows that this simplification cannot be
further improved. The vector wave equation plays a role relative to the
Maxwell equation that is analogous to the Lichnerowicz equation relative
to the linearized Einstein equations. In a future work, we will further
build on the results of~\cite{berndtson} to apply to the Lichnerowicz
equation the same simplifications as were applied to the vector wave
equation in~\cite{kh-vwtriang}. The methods developed in this work, will
help decide how much these simplifications could be improved. Of course,
it will also be very interesting to see how much the simplifications
studied jointly in~\cite{kh-vwtriang} and the current work will
translate from the (non-rotating) Schwarzschild black hole to the
significantly more complicated case of the (rotating) Kerr black hole.

An interesting generalization of the question of the existence of
rational solutions to the rational ODE $e[u] = v$ is the
characterization of the image of $e$ when applied to arbitrary rational
arguments. An equivalent question is the characterization of the
rational cokernel of $e$. Then, even if no rational solution to $e[u] =
v$ exists, precisely identifying the equivalence class of $v$ in the
cokernel of $e$ might allow us to choose a representative from the
equivalence class of $v$ that is simplest, with respect to some
reasonable criteria. Such questions also have connections with the
theory of $\mathcal{D}$-modules with rational
coefficients~\cite[Ch.2]{vanderput-singer}, \cite[Sec.10.5]{seiler},
which is an algebraic formalism for studying linear differential
equations, especially those with polynomial or rational coefficients.
These topics may also be explored in future work.

\section*{Acknowledgments}
Research of the author was partially supported by the GA\v{C}R project
18-07776S and RVO: 67985840. The author also thanks Francesco Bussola
for help with converting Equations~(3.49--51) of~\cite{berndtson} into
the form given in Examples~\ref{itm:berndtson-s02}
and~\ref{itm:berndtson-s12}.

\bibliographystyle{utphys-alpha}
\bibliography{paper-grprop}

\providecommand{\href}[2]{#2}\begingroup\raggedright\begin{thebibliography}{10}

\bibitem{abramov89b}
S.~Abramov, ``Rational solutions of linear differential and difference
  equations with polynomial coefficients,''
  \href{http://dx.doi.org/10.1016/S0041-5553(89)80002-3}{{\em USSR
  Computational Mathematics and Mathematical Physics} {\bfseries 29} (1989)
  7--12}. \url{http://mi.mathnet.ru/zvmmf3350}.

\bibitem{abramov89a}
S.~A. Abramov, ``Problems in computer algebra that are connected with a search
  for polynomial solutions of linear differential and difference equations,''
  {\em Vestnik Moskovskogo Universiteta. Seriya XV. Vychislitel'naya Matematika
  i Kibernetika} no.~3, (1989) 56--60.

\bibitem{abramov}
S.~A. Abramov, {\em Elements of the computer algebra of linear ordinary
  differential, difference and $q$-difference operators}.
\newblock MCCME, Moscow, 2014.
\newblock In Russian.

\bibitem{abramov14}
S.~A. Abramov, ``Search of rational solutions to differential and difference
  systems by means of formal series,''
  \href{http://dx.doi.org/10.1134/s0361768815020024}{{\em Programming and
  Computer Software} {\bfseries 41} (2015) 65--73}.

\bibitem{abp95}
S.~A. Abramov, M.~Bronstein, and M.~Petkov\v{s}ek,
  \href{http://dx.doi.org/10.1145/220346.220384}{``On polynomial solutions of
  linear operator equations,''} in {\em Proceedings of the 1995 international
  symposium on Symbolic and algebraic computation - ISSAC '95}, pp.~290--296.
\newblock ACM Press, New York, 1995.

\bibitem{berndtson}
M.~V. Berndtson, {\em Harmonic gauge perturbations of the Schwarzschild
  metric}.
\newblock PhD thesis, University of Colorado, 2007.
\newblock \href{http://arxiv.org/abs/0904.0033}{{\ttfamily arXiv:0904.0033}}.

\bibitem{kh-vwtriang}
I.~Khavkine, ``Explicit triangular decoupling of the separated vector wave
  equation on {Schwarzschild} into scalar {Regge-Wheeler} equations,'' 2017.
\newblock \href{http://arxiv.org/abs/1711.00585}{{\ttfamily arXiv:1711.00585}}.

\bibitem{rosa-dolan}
J.~G. Rosa and S.~R. Dolan, ``Massive vector fields on the {Schwarzschild}
  spacetime: quasinormal modes and bound states,''
  \href{http://dx.doi.org/10.1103/physrevd.85.044043}{{\em Physical Review D}
  {\bfseries 85} (2012) 044043},
  \href{http://arxiv.org/abs/1110.4494}{{\ttfamily arXiv:1110.4494}}.

\bibitem{seiler}
W.~M. Seiler, {\em Involution: The Formal Theory of Differential Equations and
  its Applications in Computer Algebra}, vol.~24 of {\em Algorithms and
  Computation in Mathematics}.
\newblock Springer, Berlin, 2010.

\bibitem{vanderput-singer}
M.~van~der Put and M.~F. Singer,
  \href{http://dx.doi.org/10.1007/978-3-642-55750-7}{{\em Galois Theory of
  Linear Differential Equations}}, vol.~328 of {\em Grundlehren der
  mathematischen Wissenschaften}.
\newblock Springer, Berlin, 2003.

\bibitem{wasow}
W.~Wasow, {\em Asymptotic Expansions for Ordinary Differential Equations},
  vol.~XIV of {\em Pure and Applied Mathematics}.
\newblock Interscience, New York, NY, 1965.

\end{thebibliography}\endgroup

\end{document}